\documentclass[letterpaper, 10 pt, journal, twoside]{ieeeconf}  
\IEEEoverridecommandlockouts
\usepackage{graphics}
\usepackage{epsfig}
\usepackage{amsmath}
\usepackage{amssymb}

\usepackage{amsthm}
\usepackage{mathabx}
\usepackage{algorithm}
\usepackage{algpseudocode}
\usepackage{url}
\usepackage{booktabs}
\usepackage{xcolor}

\newtheorem{theorem}{Theorem}[section]

\theoremstyle{definition}
\newtheorem{definition}[theorem]{Definition}
\newtheorem{assumption}[theorem]{Assumption}

\theoremstyle{remark}
\newtheorem{remark}[theorem]{Remark}

\title{\LARGE \bf
Distributed Predictive Control Barrier Functions: Towards Scalable Safety Certification in Modular Multi-Agent Systems
}

\author{Jonas Ohnemus$^{1}$, Alexandre Didier$^{1}$, Ahmed Aboudonia$^{2}$, Andrea Carron$^{1}$, and Melanie N. Zeilinger$^{1}$
\thanks{This work was supported by the Swiss National Science Foundation under grant agreement P500PT\_230223.}
\thanks{$^{1}$Institute for Dynamic Systems and Control, ETH Zurich, 8092 Zurich, Switzerland,        {\tt\small \{johnemus, adidier, carrona, mzeilinger\}@ethz.ch}}
\thanks{$^{2}$University of California Berkeley, 94720 CA, USA {\tt\small aboudonia@berkeley.edu}}%
}

\allowdisplaybreaks

\begin{document}

\maketitle
\thispagestyle{empty}
\pagestyle{empty}

%
%
%
%
%
\begin{abstract}
We consider safety-critical multi-agent systems with distributed control architectures and potentially varying network topologies. While learning-based distributed control enables scalability and high performance, a lack of formal safety guarantees in the face of unforeseen disturbances and unsafe network topology changes may lead to system failure.
To address this challenge, we introduce \emph{structured control barrier functions (s-CBFs)} as a multi-agent safety framework. The s-CBFs are augmented to a \emph{distributed predictive control barrier function (D-PCBF)}, a predictive, optimization-based safety layer that uses model predictions to guarantee recoverable safety at all times. 
The proposed approach enables a permissive yet formal plug-and-play protocol, allowing agents to join or leave the network while ensuring safety recovery if a change in network topology requires temporarily unsafe behavior. We validate the formulation through simulations and real-time experiments of a miniature race-car platoon.
\end{abstract}

\section{INTRODUCTION}
\label{sec:intro}
%
%
%
%
Modern robotic systems, such as connected autonomous vehicle fleets~\cite{caruntuDistributedModelPredictive2016}, aerial swarms~\cite{dingDistributedMachineLearning2024}, and collaborative industrial manipulators~\cite{wuDistributedCooperativeControl2023}, are increasing in scale and complexity, necessitating appropriate control architectures~\cite{scattoliniArchitecturesDistributedHierarchical2009}. As these systems scale, centralized control becomes computationally intractable, motivating a shift toward distributed architectures~\cite{geDistributedNetworkedControl2017}.
While data-driven methods such as multi-agent reinforcement learning have emerged as powerful tools for achieving high-level performance in these applications~\cite{tangDeepReinforcementLearning2024}\cite{zhangMultiAgentReinforcementLearning2021}, the learned policies generally lack formal safety guarantees and are difficult to certify. As a result, these policies are unsuitable for direct deployment in safety-critical domains.

To bridge the gap between safety and performance, \emph{safety filters} offer a modular solution~\cite{hsuSafetyFilterUnified2024}. These filters minimally modify a potentially unsafe candidate input (e.g., from a learning-based controller) to ensure safety. By leveraging model predictions, \emph{predictive safety filters} have successfully enabled aggressive maneuvers, by enforcing safety constraints only when necessary~\cite{tearlePredictiveSafetyFilter2021}. Predictive safety filters, however, merely enforce positive invariance of the safe constraint set, and the optimization may become infeasible due to disturbances that drive the state outside this safe set~\cite{wabersichLinearModelPredictive2018}. A common remedy is to relax constraints via soft-penalty formulations~\cite{kerriganSoftConstraintsExact2000}. While this approach improves feasibility, it sacrifices formal guarantees. Soft-constrained methods lack a rigorous mechanism to ensure the system eventually returns to the safe set, possibly leading to indefinite constraint violations. Recently, \emph{predictive control barrier functions} (PCBFs) have been proposed in the single-agent setting to extend the feasible region while preserving both invariance (safety) and convergence (recovery) properties~\cite{wabersichPredictiveControlBarrier2023}.

While PCBFs offer a robust solution for individual agents, maintaining these benefits in multi-agent systems with distributed learning-based control architectures remains a significant challenge. Issues of infeasibility are further compounded by the modularity of modern multi-agent systems, in which subsystems may join, leave, or reconfigure during operation. Classical approaches address the challenge of varying topology by steering the system to a specific steady state before allowing reconfiguration~\cite{zeilingerPlugPlayDistributed2013}. However, this approach is time-consuming and diverts the system from its primary operation. On the other hand, existing distributed predictive safety filters~\cite{larsenSafeLearningDistributed2017, carronDistributedSafeLearning2021, muntwilerDistributedModelPredictive2020} or controllers~\cite{conteDistributedSynthesisStability2016} may fail during modular network reconfiguration because the corresponding optimization problems become infeasible.

These limitations highlight a need for recoverable safety: safety notions that not only ensure invariance of the safe set but also guarantee convergence back to it. To address this open challenge, we propose a novel PCBF-based approach for multi-agent systems.
\begin{figure}[t]
    \centering
    \includegraphics[width=.85\linewidth]{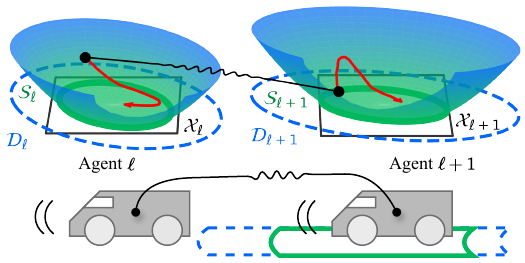}
    \vspace{-0.5em}
    \caption{Network-level recoverable safety using s-CBFs. The formulation defines safe sets $\mathcal{S}_\ell$ and regions of attraction $\mathcal{D}_\ell\setminus \mathcal{S}_\ell$, ensuring constraint satisfaction and network-level convergence even when individual trajectories temporarily have to violate safety.}
    \label{fig:sCBF_principle}
    \vspace{-1.3em}
\end{figure}
We address the safety-performance gap in multi-agent systems by developing a \emph{distributed predictive control barrier function (D-PCBF)} framework. Our approach certifies both invariance and convergence using only local information exchange. The specific contributions are: 1) We introduce \emph{structured CBFs (s-CBFs)} (see Fig.~\ref{fig:sCBF_principle}), a novel formulation that enables permissive global safety recovery using only local neighbor information. 2) We effectively enlarge the offline computed safe set and region of attraction of the s-CBFs using an online D-PCBF distributed optimization that serves as a less conservative safety and recovery certificate. 3) We design a \emph{Plug-and-Play (PnP) protocol} that leverages the slack variables from the D-PCBF to certify the admission or removal of agents online without requiring the system to disrupt its operation. 4) We validate the framework in simulation and on a \emph{real-time miniature vehicle platoon}, demonstrating that our D-PCBF successfully filters unsafe inputs from aggressive controllers and ensures collision-free recovery during high-speed plug-and-play maneuvers where existing approaches would fail.

Section~\ref{sec:prob_form} details the general problem formulation. Section~\ref{sec:dpcbfs} introduces s-CBFs and defines the D-PCBF, before Section~\ref{sec:pnp} outlines the PnP protocol. After the synthesis of s-CBFs for linear systems is described in Section~\ref{sec:sCBFs_synthesis}, simulation and real-time experimental results are presented in Section~\ref{sec:results}. Finally, Section~\ref{sec:conclusion} concludes the paper.

\section{PROBLEM FORMULATION}
\label{sec:prob_form}
We consider multi-agent dynamical systems that can be decomposed into individual dynamically coupled subsystems (or agents). We denote the index set of subsystems by $\mathcal{L}$ and the number of subsystems by $L$. The global nonlinear discrete-time dynamics are
\begin{equation}
    x(t+1)  = f\left(x(t),\, u(t)\right),\; t\in\mathbb{N}_0, \label{eq:system}    
\end{equation}
where $x\in\mathbb{R}^n$ and $u\in\mathbb{R}^m$ are the global state and input vectors, $t$ denotes time and $f:\mathbb{R}^n\times \mathbb{R}^m \to \mathbb{R}^n$ is a continuous function. We further assume that the local dynamics of the $\ell$-th subsystem with $\ell\in\mathcal{L}$ are given by
\begin{equation}
    x_\ell(t+1)  = f_\ell\left(x_{\mathcal{N}_\ell}(t),\, u_\ell(t)\right) \quad \forall \ell \in\mathcal{L}, \label{eq:dist_system}
\end{equation}
where $x_\ell \in\mathbb{R}^{n_\ell}$ and $u_\ell\in\mathbb{R}^{m_\ell}$ are the local state and input vectors, $f_\ell:\mathbb{R}^{n_{\mathcal{N}_\ell}} \times \mathbb{R}^{m_\ell}\to\mathbb{R}^{n_\ell}$ is a continuous function, and the set $\mathcal{N}_\ell$ denotes the indices of neighboring subsystems including the $\ell$-th subsystem itself, i.e., $\ell\in\mathcal{N}_\ell$. The vector $x_{\mathcal{N}_\ell}= \mathrm{col}_{j\in\mathcal{N}_{\ell}}(x_j)\in\mathbb{R}^{n_{\mathcal{N}_\ell}}$ denotes the concatenation of neighbor states. The global state and input are the concatenation of all local states and inputs, respectively, i.e., $x = \mathrm{col}_{\ell\in\mathcal{L}}(x_\ell)$ and $u = \mathrm{col}_{\ell\in\mathcal{L}}(u_\ell)$. We also assume that communication is bidirectional between neighboring subsystems and that the communication graph is connected. The $\ell$-th subsystem is subject to the local constraints
\begin{align}
    u_\ell & \in\mathcal{U}_\ell \subset \mathbb{R}^{m_\ell}, \; \forall \ell \in\mathcal{L}, \nonumber\\
    x_\ell & \in\mathcal{X}_\ell :=\{ x_\ell \mid c_\ell(x_\ell) \leq 0 \} \subset \mathbb{R}^{n_\ell}, \; \forall \ell \in\mathcal{L}, \nonumber
\end{align}
where $\mathcal{U}_\ell$ and $\mathcal{X}_\ell$ are compact sets, and $c_\ell:\mathbb{R}^{n_\ell}\to\mathbb{R}^{n_{x,\ell}}$ is a continuous function. The global constraint sets are given by $\mathcal{X}=\mathcal{X}_1 \times \ldots \times \mathcal{X}_L$ and $\mathcal{U}=\mathcal{U}_1 \times \ldots \times \mathcal{U}_L$.

We are interested in the recoverable safety of dynamical systems, which is formally defined by the existence of a safe set of states $\mathcal{S}\subseteq \mathcal{X}$ that is control invariant and can be rendered attractive from a larger region $\mathcal{D}\setminus\mathcal{S}$, where $\mathcal{S}\subset\mathcal{D}$. The following definition characterizes the safe set as the sublevel set of a continuous function $h(\cdot)$.
\begin{definition}[Control barrier function (CBF) {\cite[Def. III.1]{wabersichPredictiveControlBarrier2023}}]
    A function $h:\mathcal{D}\rightarrow \mathbb{R}$ is a CBF for system \eqref{eq:system} with safe set $\mathcal{S}=\{ x \in \mathbb{R}^n \; \vert \; h(x) \leq 0 \}\subseteq\mathcal{X}$ where $\mathcal{S}\subset \mathcal{D}$ if $\mathcal{S}$ and $\mathcal{D}$ are non-empty and compact, $h$ is continuous on $\mathcal{D}$, and there exists a continuous function $\Delta h:\mathcal{D} \rightarrow \mathbb{R}$ with $\Delta h(x) > 0$ for all $x \in \mathcal{D} \setminus \mathcal{S}$ such that
    \begin{align}
        \forall x \in \mathcal{D} \setminus \mathcal{S}:& \quad \inf_{u\in\mathcal{U}} \; h(f(x,u)) - h(x) \leq - \Delta h(x), \nonumber \\
        \forall x \in \mathcal{S}:& \quad \inf_{u\in\mathcal{U}} \; h(f(x,u)) \leq 0. \nonumber
    \end{align}
    \label{def:dtcbf}
    \vspace{-1em}
\end{definition}
CBFs thus encode recoverable safety: While $h(x)\leq 0$ implies that the constraints $\mathcal{X}$ can be satisfied via the invariance of $\mathcal{S}\subseteq \mathcal{X}$, they also provide a quantitative violation measure whenever $h(x)>0$ and a recovery guarantee that the system converges to $\mathcal S$ from $\mathcal D \setminus \mathcal S$ \cite[Theorem III.4]{wabersichPredictiveControlBarrier2023}. 
While $\mathcal{D}$ is the domain of $h$, $\mathcal{D}\setminus\mathcal{S}$ is the region of attraction. Based on the safety requirements in Definition~\ref{def:dtcbf}, one can design a safety filter~\cite{hsuSafetyFilterUnified2024} that projects the input $u_\mathrm{p}$ proposed by any controller providing system performance onto the set of safe inputs. Formally, the safety filter is given by
\begin{align}
    u(t) =\ \arg&\min_{u\in\mathcal{U}} \|{u - u_\mathrm{p}(t)} \| \; \label{eq:safety_filter} \\
    \text{s.t. }& h(f(x(t),u)) \leq 0 \text{ if } x(t) \in \mathcal S \nonumber \\
    & h(f(x(t),u)) \leq h(x(t)) -\Delta h(x(t)) \nonumber \text{ else,} \nonumber
\end{align}
where the constraints encode Definition~\ref{def:dtcbf}.

This paper aims to develop a network-level CBF that ensures recoverable overall system safety, as defined in Definition~\ref{def:dtcbf}, without imposing excessive local restrictions. For example, requiring a monotonic decrease $h_\ell(f_\ell(x_{\mathcal{N}_\ell},u_\ell)) - h_\ell(x_\ell) \leq - \Delta h_\ell(x_\ell)$ in local CBFs $h_\ell(\cdot)$ for each subsystem $\ell\in\mathcal{L}$ individually would yield the network-level CBF $h(x)=\sum h_\ell(x_\ell)$, but is often prohibitive in distributed control \cite{conteDistributedSynthesisStability2016}. Such a constraint prevents scenarios in which a safe agent must temporarily sacrifice safety to enable the recovery of a coupled unsafe agent. Consequently, we propose s-CBFs, which allow local increases in constraint violation while maintaining a strictly decreasing network-level CBF, and integrate them in a predictive optimization problem to reduce design conservatism.

\section{DISTRIBUTED PREDICTIVE CONTROL BARRIER FUNCTIONS (D-PCBFs)}
\label{sec:dpcbfs}
This section develops the theoretical D-PCBF framework. We first introduce s-CBFs to allow local relaxations of safety conditions, and then embed them as constraints in a distributed optimization problem to enlarge the safe region.
\subsection{Structured Control Barrier Functions}
\label{sec:sCBFs}
Building on the previous arguments that motivate relaxation of local decrease conditions, we formally define s-CBFs as follows. For notational convenience, we write $\bigtimes_{\ell\in\mathcal{L}} \mathcal{C}_\ell = \mathcal{C}_1\times \ldots \times \mathcal{C}_L$ for a Cartesian product of sets defined by individual sets $\mathcal{C}_\ell$ for $\ell \in\{1,\ldots L\} =\mathcal{L}$.
\begin{definition}[Structured CBFs (s-CBFs)] \label{def:structured_cbf}
    A set of local continuous functions $\{h_{\ell}\}_{\ell=1}^L$, $h_{\ell}:\mathcal{D}_{\ell} \to \mathbb{R}_{\geq0}$, is a set of s-CBFs for system \eqref{eq:dist_system} with local safe sets $\mathcal{S}_{\ell} := \{ x_\ell \;\vert \; h_{\ell}(x_\ell) = 0\}$ and domains $\mathcal{D}_{\ell} := \{ x_\ell \;\vert \; h_{\ell}(x_\ell) \leq \gamma_{f}\}$ (where $\gamma_f > 0$ and $\mathcal{D}_{\ell}$, $\mathcal{S}_{\ell}$ non-empty \& compact), if there exist continuous decrease functions $\Delta h_{\ell}:\mathcal{D}_{\ell}\to\mathbb{R}_{\geq0}$ satisfying $\Delta h_{\ell}(x_\ell)>0$ for all $x_\ell \in \mathcal{D}_{\ell}\setminus\mathcal{S}_{\ell}$ and relaxation functions $\beta_\ell: \bigtimes_{j\in\mathcal{N}_\ell} \mathcal{D}_j \to \mathbb{R}$ such that for all $\ell \in \mathcal{L}$:
    \begin{subequations}\label{eq:scbf_all}
    \begin{align}
        \forall x_{\mathcal{N}_\ell} \in \bigtimes_{j\in \mathcal{N}_\ell} \mathcal{D}_{j}:& \ \inf_{u_\ell \in \mathcal{U}_\ell} h_{\ell}(f_\ell (x_{\mathcal{N}_\ell},\, u_\ell)) - h_{\ell}(x_\ell) \nonumber\\
        & \quad \leq -\Delta h_{\ell}(x_{\ell}) + \beta_{\ell}(x_{\mathcal{N}_\ell}) \label{eq:scbf_relaxed_decrease} \\
        \forall x_{\mathcal{N}_\ell} \in \bigtimes_{j\in \mathcal{N}_\ell} \mathcal{S}_{j}:& \ \inf_{u_\ell \in \mathcal{U}_\ell} h_{\ell}(f_\ell (x_{\mathcal{N}_\ell},\, u_\ell)) = 0 \label{eq:scbf_safe_invariance} \\
        \forall x\in \bigtimes_{j\in\mathcal{L}} \mathcal{D}_{j}:& \ \sum_{\ell=1}^L \beta_{\ell}(x_{\mathcal{N}_\ell}) \leq 0. \quad \label{eq:scbf_overall_decrease}
    \end{align}
    \end{subequations}
\end{definition}
Following Definition~\ref{def:structured_cbf}, the individual s-CBFs are not local CBFs themselves, as the relaxation functions $\beta_\ell$ allow for an increase in s-CBF values. In this case, individual agents can increase their s-CBF values if necessary so that the network-level system eventually becomes safe.

Although the individual s-CBFs are not CBFs for the respective agents, we show that the network-level candidate
\begin{equation}
    h(x) = \sum_{\ell\in\mathcal{L}} h_\ell(x_\ell) \label{eq:candidate_cbf_global}
\end{equation}
is indeed a CBF according to Definition~\ref{def:dtcbf}. To this end, we make the following assumption.
\begin{assumption}[Initial s-CBF upper bound]
\label{ass:ub_levels_structured_cbfs}
    For the initial state $x(0)$ it holds that $\sum_{\ell\in\mathcal{L}} h_\ell(x_\ell(0)) \leq \gamma_f$.
\end{assumption}
This assumption is common in CBF literature and implies that the initial system state is within the region of attraction of the CBF. While this may be restrictive for conservative s-CBF designs that yield small $\gamma_f$, we show that predictive optimization can alleviate this issue. Under Assumption~\ref{ass:ub_levels_structured_cbfs}, we can show that the system state will remain within the global CBF domain for all times.

This leads to the following theorem, which proves that the network-level candidate CBF in \eqref{eq:candidate_cbf_global} is indeed a CBF according to Definition~\ref{def:dtcbf}, defining a control invariant safe set $\mathcal{S}^\mathrm{glob}$ and region of attraction $\mathcal{D}^\mathrm{glob} \setminus \mathcal{S}^\mathrm{glob}$ as follows.
\begin{theorem}[Network-level CBF from s-CBFs]
    \label{thm:global_cbf_from_structured_cbfs}
    Under Assumption~\ref{ass:ub_levels_structured_cbfs}, the function $h(x) = \sum_{\ell=1}^L h_{\ell}(x_\ell)$ constructed with s-CBFs $\{h_{\ell}\}_{\ell\in\mathcal{L}}$ according to Definition~\ref{def:structured_cbf} is a CBF according to Definition~\ref{def:dtcbf} with safe set $\mathcal{S}^{\mathrm{glob}} = \{x \; \vert \; h(x) \leq 0 \}=\bigtimes_{\ell\in\mathcal{L}} \mathcal{S}_\ell$ and domain $\mathcal{D}^\mathrm{glob} = \{ x \; \vert \;  h(x) \leq \gamma_f\}$.
\end{theorem}
\begin{proof}
    The claim follows from summing the local s-CBF inequalities: the relaxation terms satisfy $\sum_\ell \beta_\ell \le 0$, yielding decrease of the global candidate $h(x)=\sum_\ell h_\ell(x_\ell)$, while \eqref{eq:scbf_safe_invariance} gives invariance of the zero-level safe set.
\end{proof}
\subsection{Distributed Predictive Control Barrier Function}
\label{sec:DPCBF}
\begin{figure}[tbp]
    \centering
    \includegraphics[width=.8\linewidth]{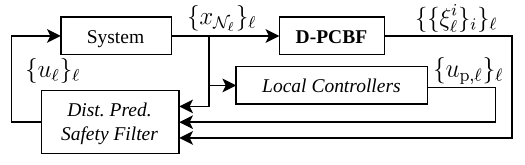}
    \caption{Possible distributed predictive control barrier function (D-PCBF) control loop. Local control policies propose inputs $u_{\mathrm{p},\ell}$ that a distributed model predictive safety filter (see \cite{muntwilerDistributedModelPredictive2020}, but with the same constraint formulation as the D-PCBF) subsequently filters to generate the control inputs $u_\ell$, using the slack values $\xi_\ell^i$ that result from the D-PCBF. The D-PCBF design ensures that the feasible set is forward-invariant and attractive. Instead of using local controllers and a distributed predictive safety filter, distributed MPC can be employed directly by modifying the cost in \eqref{eq:mpsf_optimization_problem}.}
    \label{fig:dpcbf_control_loop_arbitrary_policies}
    \vspace{-1em}
\end{figure}
To enlarge the domain $\mathcal{D}^\mathrm{glob}$ and the safe set $\mathcal{S}^\mathrm{glob}$ of s-CBFs, we define D-PCBFs. By using a model to predict state evolution towards the safe set of the s-CBFs, we can enlarge the effective safe set to be $\mathcal{S}_\mathrm{DPB}\supset\mathcal{S}^\mathrm{glob}$ and the effective domain to be $\mathcal{D}_\mathrm{DPB}\supset \mathcal{D}^\mathrm{glob}$. D-PCBFs thus enable the stabilization of a larger safe set while ensuring its positive invariance under any proposed control policy. This objective is achieved by solving the D-PCBF optimization in parallel to a distributed learning-based controller (e.g., multi-agent reinforcement learning) and using the resulting optimal slack values in a distributed predictive safety filter \cite{muntwilerDistributedModelPredictive2020}. The proposed architecture is shown in Figure~\ref{fig:dpcbf_control_loop_arbitrary_policies}. At each time step $t$, the D-PCBF optimization problem (presented here in its most general form, including the nonlinear dynamics and the general constraints of Section~\ref{sec:prob_form}) is solved as
\begin{subequations}
    \label{eq:dpcbf_optimization_problem}
    \begin{align}
    h_\mathrm{DPB}(x(t))
        &= \min_{\{u_{\ell}^i,\,\xi_{\ell}^i,\,x_{\ell}^i\}}
           \ \sum_{\ell=1}^L \left( \alpha_f \,\xi_\ell^N + \sum_{i=0}^{N-1}\! || \xi_{\ell}^i|| \right)
           \label{eq:dpcbf_obj} \\
    \text{s.t. $\forall\,\ell\in\mathcal{L}$:} \nonumber \\
    x^0_\ell &= x_\ell(t) \label{eq:dpcbf_init}\\
    x_\ell^{i+1} &= f_\ell\big( x_{\mathcal{N}_\ell}^i,\, u_\ell^i\big)
        \quad \forall\, i \in \{0,\dots,N\!-\!1\}
        \label{eq:dpcbf_dyn}\\
    x_\ell^i &\in \mathcal{X}_\ell^i(\xi_\ell^i),
        \quad \forall\, i \in \{0,\dots,N\!-\!1\}
        \label{eq:dpcbf_stage_state}\\
    x_\ell^N &\in \mathcal{S}_{f,\ell}(\xi_\ell^N)
        \label{eq:dpcbf_terminal_state}\\
    u_\ell^i &\in \mathcal{U}_\ell
        \quad \forall\, i \in \{0,\dots,N\!-\!1\}
        \label{eq:dpcbf_input}\\
    \xi_\ell^i &\ge 0
        \quad \forall\, i \in \{0,\dots,N\}.
        \label{eq:dpcbf_slack_nonneg}
    \end{align}
\end{subequations}
Furthermore, the distributed safety filter is given by
\begin{equation}
    \label{eq:mpsf_optimization_problem}
    \begin{aligned}
    \{u_\ell(t)\}_{\ell\in\mathcal{L}}
        &= \arg\min_{\{u_{\ell}^i ,\,x_{\ell}^i \}}
            \sum_{\ell=1}^L || u_{\ell}^0 - u_{\ell}^\mathrm{p}(t)|| \\
    \text{s.t. $\forall\,\ell\in\mathcal{L}$:}\;\;& \eqref{eq:dpcbf_init}, \eqref{eq:dpcbf_dyn}, \eqref{eq:dpcbf_input}\\
    &\eqref{eq:dpcbf_stage_state} \text{ \& } \eqref{eq:dpcbf_terminal_state} \text{ with } \xi_\ell^i = (\xi_\ell^i)^\star.
    \end{aligned}
\end{equation}
\begin{remark}[Multi-objective formulation] 
    Instead of the two-stage optimization, i.e., \eqref{eq:dpcbf_optimization_problem} and \eqref{eq:mpsf_optimization_problem}, a single-stage formulation could enforce a suboptimal MPC-style decrease constraint using a valid warm-start~\cite{didierMultiobjectiveApproachRobust2025}.
\end{remark}
Here, the objective \eqref{eq:dpcbf_obj} is a sum of slacks over all agents, representing the global constraint violation of the multi-agent system over the prediction horizon. As is typical in MPC, the initial state parametrizes the optimization problem via \eqref{eq:dpcbf_init}, and the dynamics \eqref{eq:dpcbf_dyn} propagate until the end of the prediction horizon $N$ based on the input sequence $\{u_\ell^0,\ldots u_{\ell}^{N-1}\}$ for all $\ell\in\mathcal{L}$. For each agent $\ell\in\mathcal{L}$, the set $\mathcal{X}_\ell^i(\xi_\ell^i) = \Big\{ x_\ell \in \mathbb{R}^{n_\ell}\ \big|\  c_\ell(x_\ell) \le \xi_\ell^i - \Delta_i\,\mathbf{1}_{n_{x,\ell}} \Big\}$ represents tightened, slacked state constraints \eqref{eq:dpcbf_stage_state}. The tightening $\Delta_0=0,\Delta_{i+1}>\Delta_i$ is mainly technical and can be chosen arbitrarily small without loss of convergence at the cost of its rate. The set $\mathcal{S}_{f,\ell}(\xi_\ell^N)=\Big\{ x_\ell \in \mathbb{R}^{n_\ell}\ \big|\ h_{f,\ell}(x_\ell) \le \xi_\ell^N \Big\}$ is defined via the function $h_{f,\ell}$ and represents the slacked terminal constraint \eqref{eq:dpcbf_terminal_state}. The slack values are constrained to be nonnegative \eqref{eq:dpcbf_slack_nonneg}. The inputs along the prediction horizon have to satisfy the decoupled input constraints \eqref{eq:dpcbf_input}. From now on, we assume that the terminal constraint functions $h_{f,\ell}$ form a set of structured CBFs according to Definition~\ref{def:structured_cbf}.
\begin{assumption}[Terminal constraints of the D-PCBF]
\label{ass:term_con_structured_cbfs}
    The terminal constraint functions $h_{f,\ell}$ in the optimization problem \eqref{eq:dpcbf_optimization_problem} are a set of s-CBFs according to Definition~\ref{def:structured_cbf} and the corresponding safe set satisfies $\bigtimes_{\ell\in\mathcal{L}} \mathcal{S}_{f,\ell}(0) \subset \bigtimes_{\ell\in\mathcal{L}} \mathcal{X}_\ell^{N-1}(0)$.
\end{assumption}
This construction allows us to prove that the optimal value function of the distributed optimization \eqref{eq:dpcbf_optimization_problem} is a CBF according to Definition~\ref{def:dtcbf}. The following theorem formalizes the theoretical properties of the proposed distributed architecture and provides the desired guarantees for recoverable closed-loop safety.
\begin{theorem}[D-PCBF]
    \label{thm:dpcbf}
    Consider a distributed PCBF $h_\mathrm{DPB}$ as defined by \eqref{eq:dpcbf_optimization_problem} and assume that $\mathcal{U}=\mathcal{U}_1\times \cdots \times \mathcal{U}_L$ and $\mathcal{X}^0(\xi)=\mathcal{X}_1^0(\xi_1) \times \cdots \times \mathcal{X}_L^0(\xi_L)$ are compact for all finite $\xi\geq 0$. Under Assumption~\ref{ass:term_con_structured_cbfs}, the optimization \eqref{eq:dpcbf_optimization_problem} admits an optimal solution, and for a finite, large enough $\alpha_f$, it follows that $h_\mathrm{DPB}$ is a CBF according to Definition~\ref{def:dtcbf} with domain $\mathcal{D}_\mathrm{DPB}:= \{ x \;\vert\; h_\mathrm{DPB}(x) \leq \alpha_f \gamma_f \}\supset \mathcal{D}^\mathrm{glob}$ and safe set $\mathcal{S}_\mathrm{DPB}:= \{ x \;\vert\; h_\mathrm{DPB}(x) =0 \}\supset \mathcal{S}^\mathrm{glob}$.
\end{theorem}
\begin{proof}
    The argument follows the PCBF construction of \cite{wabersichPredictiveControlBarrier2023}, with the terminal condition provided by the s-CBFs and the identity $(\xi_\ell^N)^\star = h_{f,\ell}((x_\ell^N)^\star)$, which places the terminal prediction in the recoverable terminal set.
    \vspace{-0.5em}
\end{proof}
Since Theorem~\ref{thm:dpcbf} certifies that the optimal value function of \eqref{eq:dpcbf_optimization_problem} is a CBF according to Definition~\ref{def:dtcbf}, the set $\mathcal{S}_\mathrm{DPB}$ is forward invariant and all trajectories that start in the region $\mathcal{D}_\mathrm{DPB}\setminus\mathcal{S}_\mathrm{DPB}$ converge to $\mathcal{S}_\mathrm{DPB}$. In fact, the proof for convergence only requires the predicted terminal state $(x^N)^\star$ to be in the terminal domain $\mathcal{D}^\mathrm{glob}$, leading to an implicitly defined domain $\{x \mid (x^N)^\star \in \mathcal{D}^\mathrm{glob} \}\supset \mathcal{D}_\mathrm{DPB}$. By performing predictive optimization online, we have effectively enlarged the safe set and the region of attraction defined by the s-CBFs. The resulting certificate, in contrast to the centralized formulation in \cite{wabersichPredictiveControlBarrier2023}, respects the multi-agent structure: state and input constraints are enforced per agent, the terminal safe set is a Cartesian product built from s-CBFs, and the only coupling appears in the dynamics \eqref{eq:dpcbf_dyn}. As a consequence, the construction is amenable to distributed optimization and preserves the communication structure.
\begin{remark}[Disturbance robustness]
    The results above assume perfect modeling. However, the CBF framework naturally yields robustness margins. If disturbances or model error are bounded, the existence of a PCBF implies input-to-state stability \cite{didierApproximatePredictiveControl2026}.
\end{remark}
%

\section{PLUG-AND-PLAY PROTOCOL}
\label{sec:pnp}
In distributed control architectures for multi-agent systems, network topological uncertainty poses a major challenge to safe operation. This section develops a certifiable PnP procedure for topology-varying multi-agent systems controlled via the D-PCBF mechanism introduced in Section~\ref{sec:dpcbfs}. Even when individual agents violate constraints immediately after a PnP operation, the proposed procedure guarantees that the network converges to a safe operating state. The protocol performs a single D-PCBF solve with re-synthesized terminal s-CBFs (for the affected neighborhood), and makes an accept/reject decision based on two certificates: (i) a violation certificate (using predetermined bounds on allowed state constraint violations) ensuring that temporary violations remain within user-specified bounds, and (ii) a recovery certificate ensuring convergence back to the safe set $\mathcal{S}_\mathrm{DPB}$.
\subsection{Plug-and-Play Setup}
We adopt the general problem setting for dynamics and constraints of Section~\ref{sec:prob_form}. Moreover, each agent defines a maximum violation set for state constraints as follows.
\begin{definition}[Maximum violation sets]
    \label{def:viol-limits}
    Each agent $\ell \in\mathcal{L}$ defines constraint violation tolerances $v_{\max,\ell}\in\mathbb R_{\ge 0}^{n_{x,\ell}}$ (limits for temporary state constraint violations) that give rise to the sets
    \begin{equation}
        \mathcal{V}_\ell := \{x_\ell \mid c_\ell(x_\ell) \leq v_{\max,\ell}\} \label{eq:max_viol_sets}
    \end{equation}
    that contain all states that are temporarily admissible.
\end{definition}
This agent-based definition of admissible constraint violations allows for individual levels of risk aversion. We assume here that the s-CBF domains are contained within the predefined maximum violation sets $\mathcal{V}_\ell$.
\begin{assumption}[Max. violation in $\mathcal{D}_{f,\ell}$]
    \label{as:max_constr_violation_in_scbf_domain}
    The maximum constraint violation in each agent's domain $\mathcal{D}_{\ell}$ satisfies
    \begin{equation}
        \max_{x_\ell \in \mathcal{D}_{\ell}} c_{\ell,k}(x_\ell) \leq v_{\max,\ell,k} \quad \forall k \in [n_{x,\ell}], \quad \forall \ell\in\mathcal{L}. \nonumber
    \end{equation}
    Recall: $n_{x,\ell}$ is the number of state constraints for agent $\ell$.
\end{assumption}
This assumption forces the recoverable s-CBF domain to respect the user-defined limits ($\mathcal{D}_\ell \subseteq \mathcal{V}_\ell$), ensuring that safety recovery never yields unacceptable violations.
Given the fixed bounds $v_{\max,\ell,k}$, the constraint in Assumption~\ref{as:max_constr_violation_in_scbf_domain} can be enforced in the synthesis of s-CBFs for linear systems in \eqref{eq:max_safeset_synthesis_opt} without loss of convexity.
\subsection{Certificates and Algorithm}
Given Assumption~\ref{as:max_constr_violation_in_scbf_domain}, we introduce Algorithm~\ref{alg:pnp} to perform the PnP procedure online. The re-synthesis step for the new, proposed network in Line~\ref{alg:pnp_resynthesis} (also see \eqref{eq:max_safeset_synthesis_opt}) can be restricted to a neighborhood of subsystems affected by the plug-in or plug-out operation, or performed offline before operation if there are only a limited number of possible network topologies and agent dynamics configurations. Solving the D-PCBF in Line~\ref{alg:pnp_dpcbf_solve} (also see \eqref{eq:dpcbf_optimization_problem}) involves solving an online distributed optimization problem that depends on the current state. The recovery certificate in Line~\ref{alg:pnp_recovery_certificate} checks whether the optimal terminal slacks are upper-bounded by $\gamma_f$, and is justified by Theorem~\ref{thm:dpcbf} to verify whether the global terminal predicted state is in the region of attraction of the terminal s-CBFs. It involves a global summation that requires a predefined global protocol but can be performed within the connected communication graph. The violation certificate in Line~\ref{alg:pnp_violation_certificate} ensures recursive feasibility, in that predicted states over the horizon will never leave the maximum violation sets \eqref{eq:max_viol_sets} for future time steps.
\begin{algorithm}[htbp]
    \caption{PnP procedure with D-PCBF certificates}
    \label{alg:pnp}
    \begin{algorithmic}[1]
    \State Request: $\mathcal{L} \rightarrow \mathcal{L}_\mathrm{new}$ (new total set of agent indices)
    \State Re-synthesize $h_{f,\ell}$ for $\ell \in \mathcal{L}_\mathrm{new}$ \label{alg:pnp_resynthesis}
    \State Solve D-PCBF for $\{(\xi_\ell^{i})^{\star}\}_{\ell\in\mathcal{L}_\mathrm{new},\,i\in\{0,...,N\}}$ \label{alg:pnp_dpcbf_solve}
    \State \textbf{Recovery Certificate:} Compute and check whether
    \[
        \sum_{\ell\in\mathcal{L}_\mathrm{new}} (\xi_\ell^{N})^{\star} \leq \gamma_f.
    \] \label{alg:pnp_recovery_certificate}
    \State \textbf{Violation Certificate:} For $\ell\in\mathcal{L}_\mathrm{new}$, $i\in[N-1]$, check 
    \[
        \max(0,\, (\xi_\ell^{i})^{\star} - \Delta_i \mathbf{1}) \leq v_{\max,\ell}.
    \] \label{alg:pnp_violation_certificate}
    \State \textbf{If} both certificates pass: accept, \textbf{else}: reject.
    \end{algorithmic}
\end{algorithm}
\begin{remark}[Distributed computation of recovery certificate]
    The global summation check in Line~\ref{alg:pnp_recovery_certificate} is required only once per change of network topology. It can be computed via standard consensus protocols or hop-by-hop propagation~\cite{lynchDistributedAlgorithms1996}. Alternatively, one may enforce conservative local bounds $(\xi_{\ell}^{N})^{\star} \le \gamma_{f}/L_\mathrm{new}$ through division by the new total number of agents $L_\mathrm{new}$ to avoid the global consensus problem.
\end{remark}

\section{IMPLEMENTATION AND RESULTS}
\label{sec:implementation_and_results}
In this section, we will first show methods of how to synthesize s-CBFs for linear distributed systems. Then, we will show results of using the D-PCBF framework in simulation and on real miniature cars.
\subsection{Synthesis of s-CBFs for Linear Systems}
\label{sec:sCBFs_synthesis}
In this section, we describe distributed methods to synthesize s-CBFs. We focus here on two constructive synthesis approaches for linear systems subject to polytopic state and input constraints: one in which the origin is the safe set, and one in which agents compute ellipsoidal safe sets. While the former approach follows from synthesizing structured control Lyapunov functions as in \cite{conteDistributedSynthesisStability2016}, it yields the smallest safe set. On the other hand, the latter yields possibly larger safe sets, but requires an iterative construction method due to nonconvexity. While it is possible to synthesize barrier functions for general nonlinear systems (e.g., via sum-of-squares techniques or learning-based methods \cite{dawsonSafeControlLearned2023}), this remains computationally intensive compared with the following scalable, efficient synthesis for linear systems.

We consider here linear time-invariant multi-agent systems, with the $\ell$-th agent's dynamics written as
\begin{equation}
    x_\ell(t+1) = A_{\mathcal{N}_\ell} x_{\mathcal{N}_\ell}(t) + B_\ell u_\ell(t) \quad \forall \ell \in\mathcal{L}, \label{eq:lin_sys_dist_synthesis}
\end{equation}
where $A_{\mathcal{N}_\ell}\in\mathbb{R}^{n_\ell\times n_{\mathcal{N}_\ell}}$, and $B_\ell \in\mathbb{R}^{n_\ell \times m_\ell}$. The state and input constraint sets are formally written for all $\ell\in\mathcal{L}$ as
\begin{align}
    \mathcal{X}_\ell = \{x_\ell \mid H^x_\ell x_\ell \leq h^x_\ell \} \ \text{ and } \ 
    \mathcal{U}_\ell = \{u_\ell \mid H^u_\ell u_\ell \leq h^u_\ell \}, \nonumber
\end{align}
where $H^x_\ell \in\mathbb{R}^{n_{x,\ell}\times n_\ell}$, $h^x_\ell \in\mathbb{R}^{n_{x,\ell}}$, $H^u_\ell \in\mathbb{R}^{n_{u,\ell}\times m_\ell}$, and $h^u_\ell \in\mathbb{R}^{n_{u,\ell}}$ for all $\ell\in\mathcal{L}$. We assume that all constraint sets contain the origin in their interior, i.e., $0\in \mathrm{int}\, \mathcal{X}_\ell$ and $0\in\mathrm{int} \, \mathcal{U}_\ell$.
Motivated by standard synthesis procedures for quadratic Lyapunov functions \cite{boydLinearMatrixInequalities1994}, the s-CBFs are parametrized as
\begin{equation}
    h_{\ell}(x_\ell) = \max(0, x_\ell^\top P_{\ell} x_\ell - \gamma_{x,\ell}) \quad \forall \ell \in\mathcal{L}, \label{eq:structured_cbf_lin_candidates}
\end{equation}
with matrices $P_{\ell}=P_{\ell}^\top \succ 0$ of appropriate dimension and nonnegative scalars $\gamma_{x,\ell}$. Given this parametrization, the local safe sets and domains follow as
\begin{align}
    \mathcal{S}_{\ell} &= \{ x_\ell \;\vert\; x_\ell^\top P_{\ell} x_\ell \leq \gamma_{x,\ell} \} \nonumber \\
    \mathcal{D}_{\ell} &= \{ x_\ell \;\vert\; x_\ell^\top P_{\ell} x_\ell \leq \gamma_{x,\ell} + \gamma_f \}. \nonumber
\end{align}
For the s-CBF synthesis, we aim to leverage linear state-feedback control policies, as is common in the design of Lyapunov functions for linear systems. Since each agent can leverage information from neighboring agents, the considered local policies are given by the matrices $K_{\mathcal{N}_\ell}$, and the corresponding control action is
\begin{equation}
    u_\ell = K_{\mathcal{N}_\ell} x_{\mathcal{N}_\ell}. \label{eq:structured_cbf_lin_control_candidate}
\end{equation}
To facilitate notation, we define selection matrices from network-level to local and neighboring states for each agent $\ell\in\mathcal{L}$ as $T_\ell \in\mathbb{R}^{n_\ell \times n}$ and $W_\ell\in\mathbb{R}^{n_{\mathcal{N}_\ell} \times n}$, respectively.
\subsubsection{Synthesis with Origin as Safe Sets}
Using the origin as the safe set is equivalent to setting $\gamma_{x,\ell}=0$ for all $\ell \in \mathcal {L}$. We choose a constant $\gamma_f = 1$, as the matrices $P_\ell$ can be rescaled, and use the function candidates $\beta_\ell(x_{\mathcal{N}_\ell}) = x_{\mathcal{N}_\ell}^\top \Gamma_{\mathcal{N}_\ell} x_{\mathcal{N}_\ell}$ and $\Delta h_\ell(x_\ell) = x_\ell^\top D_\ell x_\ell$ for all $\ell \in \mathcal{L}$ where $\Gamma_{\mathcal{N}_\ell}^\top = \Gamma_{\mathcal{N}_\ell}$ and $D_\ell^\top = D_\ell \succ 0$. To satisfy the s-CBF conditions in Definition~\ref{def:structured_cbf}, a distributed semidefinite program can be solved. This approach is similar to the synthesis of terminal control Lyapunov functions in \cite{conteDistributedSynthesisStability2016}. While \cite{conteDistributedSynthesisStability2016} synthesizes time-varying terminal safe sets, we design time-varying regions of attraction. Condition \eqref{eq:scbf_relaxed_decrease} becomes a Lyapunov decrease LMI, \eqref{eq:scbf_safe_invariance} is satisfied by construction, and \eqref{eq:scbf_overall_decrease} can be enforced via coupled structured LMIs with block-diagonal matrices $S_{\mathcal{N}_\ell}$ that bound $\Gamma_{\mathcal{N}_\ell}$. However, unlike \cite{conteDistributedSynthesisStability2016}, we only enforce input and not state constraints for all states in $\mathcal{D}_\ell$, as the recoverable region is not required to be within the state constraints $\mathcal{X}_\ell$. The synthesis optimizes for $\delta_\ell^\mathrm{o} = \{E_\ell, Y_{\mathcal{N}_\ell}, S_{\mathcal{N}_\ell}, \Gamma_{\mathcal{N}_\ell}, D_{\ell} \}$, where $E_\ell=P_\ell^{-1}$ and $Y_{\mathcal{N}_\ell}=K_{\mathcal{N}_\ell}E_{\mathcal{N}_\ell}$. The matrix $E_{\mathcal{N}_\ell}=\mathrm{blkdiag}_{j\in\mathcal{N}_\ell}(E_j)$ is simply the block-diagonal matrix of neighboring $E_j$.
\subsubsection{Synthesis with Ellipsoidal Safe Sets}
Aiming to enlarge the safe sets, we propose using the s-CBF candidates \eqref{eq:structured_cbf_lin_candidates} with $\gamma_{x,\ell}>0$. In this case, we choose $\gamma_{x,\ell}=1$ and the function candidates $\beta_\ell (x_{\mathcal{N}_\ell}) = \sum_{j\in\mathcal{N}_\ell} b_{\ell,j} h_j(x_j)$ and $\Delta h_\ell(x_\ell) = \rho_\ell h_\ell(x_\ell)$ with some constant $\rho_\ell>0$ for all $\ell\in\mathcal{L}$. For the resulting s-CBFs to satisfy the conditions of Definition~\ref{def:structured_cbf}, we compute the local decision variables $\delta_\ell^\mathrm{e} = \{E_\ell, \{b_{\ell,j}\}_{j\in\mathcal{N}_\ell}, Y_{\mathcal{N}_\ell}\}$ and the network-level decision variable $\gamma_f$ by solving the following optimization problem, where $E_\ell = P_\ell^{-1}$ and $Y_{\mathcal{N}_\ell} = K_{\mathcal{N}_\ell} E_{\mathcal{N}_\ell}$ are defined as in the previous subsection,
\begin{subequations}
    \label{eq:max_safeset_synthesis_opt}
    \begin{align}
        \min_{\{ {\delta_\ell^\mathrm{e}} \}_{\ell \in \mathcal{L}}, \gamma_f} & -\sum_{\ell=1}^L \log\det({E_\ell}) \\
        \textrm{s.t.}& \;\forall \ell\in\mathcal{L}: \nonumber \\
        & \quad {E_\ell} \succ 0,\; \rho_\ell > 0, \\
        & \quad b_{\ell,j} \geq 0 \, \forall j \in\mathcal{N}_\ell \setminus \{\ell\}, \; 1 - \rho_\ell + b_{\ell,\ell} \geq 0, \label{eq:syn_elip_relaxed_decrease_0}\\
        & \quad \sum_{j\in\mathcal{N}_\ell} b_{\ell,j} \leq \rho_\ell, \label{eq:syn_elip_relaxed_decrease_1} \\
        & \quad M_\ell := \sum_{j\in\mathcal{N}_\ell} b_{\ell,j} W_\ell T_j^\top E_j T_j W_\ell^\top, \nonumber \\
        & \quad \begin{bmatrix}
            {(1-\rho_\ell) \cdot \bar{E}_\ell} + M_\ell & {E_{\mathcal{N}_\ell}} A_{\mathcal{N}_\ell}^\top + {Y_{\mathcal{N}_\ell}}^\top B_\ell^\top \\
            A_{\mathcal{N}_\ell} {E_{\mathcal{N}_\ell}} + B_\ell {Y_{\mathcal{N}_\ell}} & {E_\ell}
        \end{bmatrix} \nonumber \\
        &\quad \quad \succeq 0, \label{eq:syn_elip_relaxed_decrease_2}\\
        & \quad \sum_{j\in\mathcal{N}_\ell} b_{j,\ell} \leq 0, \label{eq:syn_elip_globalrelax} \\
        & \quad \begin{bmatrix}
            (h^u_{\ell,i})^2/(\vert{\mathcal{N}_\ell}\vert (1+\gamma_f)) &  (H^u_{\ell,i}) {Y_{\mathcal{N}_\ell}} \\
            { Y_{\mathcal{N}_\ell}}^\top (H^u_{\ell,i})^\top & {E_{\mathcal{N}_\ell}}
        \end{bmatrix} \nonumber \\
        & \quad \quad \succeq 0 \quad \forall i \in [n_{u,\ell}], \label{eq:syn_elip_inputcon}\\
        & \quad (H^x_{\ell,i}) E_\ell  (H^x_{\ell,i})^\top \leq (h^x_{\ell,i})^2 \quad \forall i \in [n_{x,\ell}], \label{eq:syn_elip_statecon}
    \end{align}
\end{subequations}
where $\bar{E}_\ell = W_\ell T_\ell^\top E_\ell T_\ell W_\ell^\top$ and the other variables are as defined above. Inequalities \eqref{eq:syn_elip_relaxed_decrease_0}, \eqref{eq:syn_elip_relaxed_decrease_1}, and \eqref{eq:syn_elip_relaxed_decrease_2} ensure conditions \eqref{eq:scbf_relaxed_decrease} and \eqref{eq:scbf_safe_invariance}. Inequality \eqref{eq:syn_elip_globalrelax} ensures condition \eqref{eq:scbf_overall_decrease}. Finally, \eqref{eq:syn_elip_inputcon} and \eqref{eq:syn_elip_statecon} ensure state and input constraint satisfaction, respectively. The proof of sufficiency is omitted here for brevity. To solve the non-convex (bilinear) problem \eqref{eq:max_safeset_synthesis_opt} approximately, we employ a distributed block-coordinate descent (alternating minimization) scheme \cite{iwasakiXYcentringAlgorithmDual1995}. By fixing the variables $\left( \{\{b_{\ell,j}\}_{j\in\mathcal{N}_\ell}\}_\ell, \gamma_f \right)$, the problem becomes convex in $\left(\{E_\ell\}_\ell, \{Y_{\mathcal{N}_\ell}\}_\ell \right)$ and by fixing $\left(\{E_\ell\}_\ell\right)$, the problem is convex in $\left( \{\{b_{\ell,j}\}_{j\in\mathcal{N}_\ell}\}_\ell, \{K_{\mathcal{N}_\ell}\}_\ell, (1+\gamma_f)^{-1} \right)$. Crucially, because the constraints respect the network topology \cite{conteDistributedSynthesisStability2016}, each of these LMI subproblems can be solved via distributed optimization techniques such as the alternating direction method of multipliers (ADMM)~\cite{boydDistributedOptimizationStatistical2010}, enabling a fully distributed synthesis procedure.
\subsection{Simulation Results and Experiments}
\label{sec:results}

We validate the proposed D-PCBF architecture using a vehicle platooning example. The validation proceeds in three stages: s-CBF synthesis for the longitudinal model, simulations to assess scalability, and hardware experiments on miniature race cars to demonstrate real-world robustness.
\subsubsection{Car Platoon Synthesis}
\label{sec:results_platoon_synth}
\begin{figure}[htbp]
    \centering
    \includegraphics[width=.87\linewidth]{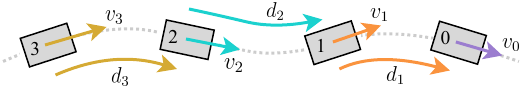}
    \vspace{-0.3em}
    \caption{Platoon example system. Assuming perfect lateral tracking and a kinematic model, the distributed dynamics consider the arc length to the preceding vehicle and the velocity magnitude.}
    \label{fig:platoon}
    \vspace{-0.7em}
\end{figure}
The platoon is modeled as a double integrator with purely longitudinal dynamics, assuming perfect lateral tracking (Fig.~\ref{fig:platoon}). The state of the $\ell$-th vehicle evolves according to discretized dynamics \cite{alamFuelefficientDistributedControl2011}, with sampling time $\mathrm{d}t=0.1\,\mathrm{s}$,
\begin{align}
    \Tilde{d}_\ell(t+1) &= \Tilde{d}_\ell(t) - \mathrm{d}t\,\Tilde{v}_\ell(t) 
    + \mathrm{d}t\,\Tilde{v}_{\ell-1}(t), \label{eq:platoon_distance_dyn}\\
    \Tilde{v}_\ell(t+1) &= \Tilde{v}_\ell(t) + \mathrm{d}t\,a_\ell(t), \label{eq:platoon_velocity_dyn}
\end{align}
where $\Tilde{d}_\ell$ and $\Tilde{v}_\ell$ are errors relative to the reference distance $d_\mathrm{ref}=1.0\,\mathrm{m}$ and velocity $v_\mathrm{ref}=1.0\,\mathrm{m}/\mathrm{s}$. 
We constrain inputs $a_\ell \in [-5.0, 5.0]\,\mathrm{m}/\mathrm{s}^2$ and states $d_\ell \in [0.5, 1.5]\,\mathrm{m}$, $v_\ell \in [0.5, 1.5]\,\mathrm{m}/\mathrm{s}$. Agent $\ell=0$ has no distance state \eqref{eq:platoon_distance_dyn}. 

We first approximately solve \eqref{eq:max_safeset_synthesis_opt} to maximize the safe set. As shown in Fig.~\ref{fig:platoon_synthesis}, the alternating optimization (cf. Section~\ref{sec:sCBFs_synthesis}) successfully expands the volume of the safe set measured by $\log\det(E)$ and the terminal region of attraction measured by $\gamma_f$.
\begin{figure}[tbp]
    \centering
    \includegraphics[width=.83\linewidth]{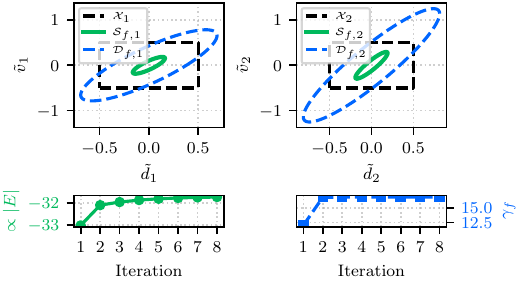}
    \vspace{-0.5em}
    \caption{Synthesis results for the vehicle platoon model when maximizing the safe set for $L=5$ agents. The top plots show the resulting s-CBF safe sets (green) and domains (blue) for the second and third car of the platoon. The bottom plots visualize the alternating optimization results described in Section~\ref{sec:sCBFs_synthesis} to approximately solve \eqref{eq:max_safeset_synthesis_opt}. The safe set size is plotted as the log-determinants of $E$, i.e., $\log\det(E)=\sum_\ell \log\det(E_\ell) \propto \vert E\vert $, while $\gamma_f$ is used to indicate the size of the terminal region of attraction.}
    \label{fig:platoon_synthesis}
    \vspace{-1em}
\end{figure}
\subsubsection{Scalability Analysis (Simulation)}
\label{sec:results_platoon_sim}
To evaluate the framework's scalability, we simulate a scenario in which an adversarial controller commands a constant, unsafe acceleration $a_\ell^\mathrm{p} = 10.0\, \mathrm{m}/\mathrm{s}^2$ for all agents. To filter these inputs, we employ the control architecture presented in Fig.~\ref{fig:dpcbf_control_loop_arbitrary_policies}. We use the origin as a safe set and synthesize the s-CBFs according to Section~\ref{sec:sCBFs_synthesis}. The platoon is initialized in an unsafe state, with vehicles 1 and 2 at zero distance to their predecessors (i.e., the cars are in contact) and all other relative states set to zero. Constraints and dynamics are the same as previously described. We set $\alpha_f=10^3$, $\Delta_i=10^{-3}i$, and $N=10$.

As shown in Fig.~\ref{fig:platoon_simulation} (top), the filter mechanism successfully intervenes, overriding the unsafe commands to recover safe spacing. We compare the solve times of distributed and centralized implementations (Fig.~\ref{fig:platoon_simulation}, bottom). MOSEK \cite{mosek} is used as the centralized solver and for the distributed agent-level problems. To solve (12) in a distributed manner, we implement a custom solver based on ADMM~\cite{boydDistributedOptimizationStatistical2010}. The problem is decomposed into local subproblems, in which agents optimize local trajectories and copies of coupled neighbor states, subject to consensus constraints. We use warm-starting to improve convergence speed, and the ADMM procedure for \eqref{eq:dpcbf_optimization_problem} converges in an average of 15 iterations. Since the distributed algorithm is executed on a sequential CPU, we evaluate it using an idealized parallel runtime, computed by summing, over ADMM iterations, the maximum local subproblem solve time and the maximum consensus and dual update times across agents, yielding a communication-free lower bound on the runtime of a fully parallel implementation. While the centralized solver's computation time grows with network size, the idealized distributed runtime remains approximately constant, indicating favorable scalability of the distributed formulation.
\begin{figure}[tbp]
    \centering
    \includegraphics[width=.77\linewidth]{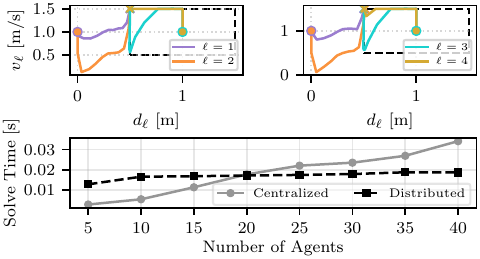}
    \caption{Simulation results for the vehicle platoon model. The top plots show the state space trajectories (start=dot, end=cross) of the distributed simulation for the first four follower agents, where agents $\ell=1,2$ are initialized at an unsafe distance. The top-left plot shows results for 5 agents, and the right plot shows results for 40 agents. In both cases, the network-level system converges to safe operation within the state constraints, although all agents try to accelerate with $a_\ell^\mathrm{p}=10.0\, \mathrm{m}/\mathrm{s}^2$. The bottom plot shows the solve times of \eqref{eq:dpcbf_optimization_problem} for the centralized solver and the distributed implementation, where the latter is reported as an idealized parallel runtime.}
    \label{fig:platoon_simulation}
    \vspace{-1em}
\end{figure}
\subsubsection{Real-time Miniature Race Car Platoon}
\label{sec:results_platoon}
In the previous simulation, the closed-loop system was entirely simulated with perfect model knowledge. To provide deeper insight into the inherent robustness of the D-PCBF formulation and to demonstrate its potential for online PnP operation, we conduct a vehicle platoon experiment within the miniature race car framework \cite{carronChronosCRSDesign2023}. A decentralized longitudinal controller proposes unsafe inputs, which are subsequently filtered using the D-PCBF and the predictive safety filter control loop shown in Figure~\ref{fig:dpcbf_control_loop_arbitrary_policies}. Local lateral pure-pursuit controllers are assumed to provide safe steering inputs that keep the cars on a given reference track. The time step is chosen to be $\mathrm{d}t=0.1\,\mathrm{s}$, and we set $v_\mathrm{ref}=0.5\mathrm{m}/\mathrm{s}$ and $d_\mathrm{ref}=1.0\mathrm{m}$. The constraints are now $a_\ell \in [-0.5, 1.5]\,\mathrm{m}/\mathrm{s}^2$, $d_\ell \in [0.8, 1.4]\,\mathrm{m}$, and $v_\ell \in [0.1, 0.9]\,\mathrm{m}/\mathrm{s}$. The leader's acceleration $a_0$ is further limited to $25\%$ of the follower's maximum acceleration to induce unsafe platoon situations by artificially limiting the leader's ability to accelerate away from the approaching follower.

To synthesize the terminal s-CBFs, we use the origin as the safe set and perform the associated synthesis optimization described in Section~\ref{sec:sCBFs_synthesis}. For the experiment, we again employ the control architecture presented in Figure~\ref{fig:dpcbf_control_loop_arbitrary_policies}. To demonstrate the framework's inherent robustness, we select an unsafe PID controller that aims to exceed the speed limit at $1.5\mathrm{m}/\mathrm{s}$, thereby creating unsafe driving conditions. This controller proposes acceleration inputs $a_\ell^\mathrm{p}$, which are filtered by the proposed D-PCBF. The D-PCBF parameters are set to $\alpha_f=10^{4}$, $N=40$, and $\Delta_i=10^{-3}i$. The closed loop with D-PCBF and safety filter runs at $100\,\mathrm{Hz}$, and we use HPIPM QCQP \cite{frisonIntroducingQuadraticallyconstrainedQuadratic2022} as a solver for both.

\begin{figure*}[tbp]
    \centering
    \includegraphics[scale=0.87]{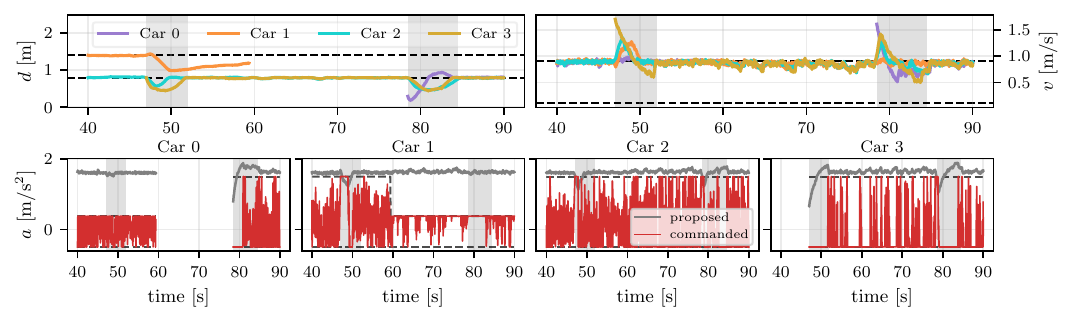}
    \vspace{-0.4cm}
    \caption{State and input signals of the platoon experiment. The experiment starts with the configuration $0\leftarrow1\leftarrow2$. It consists of the following events: Car 3 approaches with high speed from behind and plugs in at $t\approx47s$. Then, the leading car 0 plugs out at $t\approx 59s$. Finally, at $t\approx78s$, this car approaches with high speed from behind and plugs in again. It can be seen that, in both plug-in scenarios, temporary violations of velocity and distance constraints occur, since the cars have limited acceleration. Although the infeasible commands proposed by the unsafe PID controllers (gray signals) would violate state constraints, the framework prevents unsafe operation whenever possible by applying the filtered commands (red signals). Moreover, it enables safety recovery in cases of constraint violations (gray shaded times). The leader vehicle is constrained to a smaller maximum acceleration to generate visible distance-constraint violations. }
    \label{fig:platoon_experiment}
    \vspace{-1em}
\end{figure*}
We perform a high-speed join maneuver at the tail of the platoon. This setup replicates an \emph{end-of-queue scenario}, a critical safety case in traffic management \cite{khanIntelligentInfrastructurebasedQueueend2007}, in which a vehicle approaches a formed platoon with significant relative velocity, requiring the platoon to temporarily reduce relative distances below the safety threshold. A platoon consisting of cars 0, 1, and 2 is launched to drive a rectangular track with rounded corners. At some time, car 3 is launched to approach the end of the platoon at high speed. Shortly before a collision at $t\approx47s$, a plug-in is performed. Subsequently, car 0 performs a plug-out at $t\approx 59s$ and stops to the side of the track. It is then launched again, approaching the end of the platoon at high speed to perform another plug-in at $t\approx78s$, repeating car 3's earlier high-speed approach. 

The resulting state and input trajectories are shown in Figure~\ref{fig:platoon_experiment}, where it can be seen that temporary violations of the distance and velocity constraints (shaded in gray) cannot be avoided. Yet the closed-loop system converges to safe operation, indicating that selecting the origin as the safe set is sufficient for the chosen horizon length. This experiment highlights the effectiveness of the distributed formulation: While a distributed MPC \cite{conteDistributedSynthesisStability2016} or a standard predictive safety filter \cite{muntwilerDistributedModelPredictive2020} would become infeasible, the proposed approach guarantees convergence to safe operation. Furthermore, conducting these experiments with fully decentralized controllers (e.g., PID) would likely result in collisions because the approaching vehicle cannot decelerate sufficiently due to actuator constraints. Note that we disable the PnP protocol (Algorithm~\ref{alg:pnp}) and admit all PnP requests to allow extreme scenarios and demonstrate the robustness of D-PCBFs. While the recovery (Line~\ref{alg:pnp_recovery_certificate}) and violation certificate (Line~\ref{alg:pnp_violation_certificate}) would have held in both plug-in scenarios, we did not enforce Assumption~\ref{as:max_constr_violation_in_scbf_domain} in the synthesis of the s-CBFs. Although we have omitted Algorithm~\ref{alg:pnp} here, it would allow us to accommodate varying levels of risk aversion and formally guarantee safety recovery.

In addition to this experiment, which involved high-speed plug-in behavior, we have performed plug-ins from standstill and plug-outs of cars at various positions within the platoon. Videos of the experiments accompany this manuscript.\footnote{Available at \url{https://youtu.be/Anr12ZRGJuM} and\\ \url{https://youtu.be/Btexrte7vWQ}.}
\begin{remark}[Distributed implementation]
    In these hardware experiments, because of the small number of agents, centralized solvers were used to validate the \emph{formulation}. In a real-world deployment, the achievable control frequency would depend on solver convergence speed and network bandwidth, a typical trade-off in distributed optimization~\cite{tsianosCommunicationComputationTradeoffs2012}.
\end{remark}
%
%

\section{CONCLUSION}
\label{sec:conclusion}
This paper proposes a distributed safety-filter framework based on predictive control barrier functions for multi-agent systems. A distributed optimization combines local s-CBFs into a terminal constraint, and the resulting optimal value function serves as a CBF for the overall network. This construction yields network-level forward invariance and convergence to a safe set after temporary constraint violations. It further underpins a plug-and-play protocol that uses a slack-based admissibility check based solely on neighbor information, enabling safe admission and removal of agents without requiring preparatory system behavior. Validation on a simulated and a physical vehicle platoon demonstrates recovery to safety, sustained invariance under coupling, and safe plug-and-play behavior, indicating the practicality of the approach for scalable safety certification. Current limitations include qualitative treatment of robustness, potentially conservative s-CBF synthesis for linear polytopic systems, and the restriction to coupling via dynamics. Future work includes deploying distributed solvers on hardware to assess latency and network effects, designing less conservative (e.g., time-varying) terminal sets to enlarge safe regions, establishing robustness guarantees for bounded uncertainties, and extending the framework to coupled state constraints.

\vspace{-0.2cm}









\bibliography{references}
\bibliographystyle{ieeetr}

\end{document}